\documentclass[11pt]{article}

\usepackage[utf8]{inputenc}
\usepackage[margin=1in]{geometry}
\usepackage[dvipsnames]{xcolor}
\usepackage{authblk}
\usepackage[en-US]{datetime2}

\usepackage{amsmath}
\usepackage{amssymb}
\usepackage{mathtools}
\usepackage{amsthm}
\usepackage{thm-restate}   % restating results in the appendix
\usepackage{mathrsfs}
\usepackage{dsfont}
\usepackage{upgreek}
\usepackage{relsize}
\usepackage{physics}
\allowdisplaybreaks                 % let long calculations break across pages

\usepackage{enumitem}
\usepackage{booktabs, multirow, makecell}
\usepackage{graphicx}
\usepackage{subcaption}
\usepackage{float}
\usepackage{placeins}
\usepackage{capt-of}
\usepackage{tikz}

\usepackage{mdframed}
\usepackage[most]{tcolorbox}
\usepackage{algorithm}
\usepackage{algpseudocode}
\usepackage{matlab-prettifier}
\usepackage{listing}

\usepackage{comment}
\usepackage{mwe}

\usepackage[sort&compress,numbers]{natbib}
\usepackage{hyperref}

\theoremstyle{plain}
\newtheorem{remark}{Remark}
\newtheorem{theorem}{Theorem}
\newtheorem{lemma}{Lemma}
\newtheorem{proposition}{Proposition}
\newtheorem{corollary}{Corollary}

\theoremstyle{definition}
\newtheorem{definition}{Definition}

\newcommand{\id}{\mathds{1}}                        % identity operator
\newcommand{\Tra}{\mathrm{Tr}}                      % trace
\newcommand{\EE}{\mathop{\mathbb{E}}\nolimits}      % expectation (subscript, not limits)
\renewcommand{\Re}{\operatorname{Re}}
\newcommand{\arccosh}{\operatorname{arccosh}}

\newcommand{\Dsand}{\widetilde{D}_{\alpha}}                 % sandwiched Renyi divergence
\newcommand{\Hup}{\widetilde{H}^{\uparrow}_{\alpha}}        % uparrow sandwiched conditional entropy
\newcommand{\sigmaopt}{\Tilde{\sigma}^{*}_{B}(\alpha)}      % optimizer in \Hup
\newcommand{\Sig}{\Theta}                                   % reference operator 1_A \otimes \sigma_B (change symbol here)

\newcommand{\Renyi}{R\'enyi}

\DeclarePairedDelimiterX{\infdivx}[2]{(}{)}{#1\;\delimsize\|\;#2}

\newcommand{\cH}{\mathcal{H}}

\title{Fully quantum fidelity-based entropic continuity bounds}

\author[1]{Matthew Simon Tan}
\author[1]{Ernest Y.-Z.\ Tan}
\affil[1]{Department of Physics, National University of Singapore}

\date{}

\begin{document}

\maketitle

\begin{abstract}
We present a continuity bound for quantum conditional entropies in terms of fidelity, with no dependence on the dimension of the conditioning system. Our results encompass both the von Neumann entropy and one version of sandwiched Rényi entropy, which is particularly relevant for applications in quantum key distribution (QKD). In such applications, this bound should provide improved results as compared to previous continuity bounds based on trace distance. We also show that the bounds can be sharpened as expected for classical-quantum states, and verify that they hold when the conditioning system has countably infinite dimension. Finally, we present several variants tailored for exponential functions of the Rényi entropies instead, which should be relevant for QKD applications. 
\end{abstract}

\section{Introduction}

In quantum information theory, one of the most critical concepts is the notion of entropy for quantum states. Many different versions of quantum entropy have been proposed, with a wide range of important applications~\cite{tomamichel2015quantum,khatri2024principles}. Broadly, however, they can be qualitatively viewed as measures of how much ``randomness'' or ``uncertainty'' there is in a quantum state, which can then be translated into operational applications such as data compression or secret key extraction.

In particular, one major application of quantum entropies is in the field of quantum key distribution (QKD), where they quantify the rate at which secret key bits can be extracted from a QKD protocol (in short, the \emph{keyrate}). 
Basically, the achievable keyrate in the asymptotic limit of infinitely many QKD protocol rounds can be bounded in terms of the well-known 
\emph{von Neumann entropy}~\cite{NC10}, via information-theoretic methods~\cite{devetak2005distillation,rennerthesis}. Subsequently, more recent work has shown that to analyze QKD performance for realistic, finite numbers of rounds, a more directly relevant quantity is a generalized entropy family known as the \emph{{\Renyi} entropies}~\cite{dupuis2023,fawzi2026additivity,arqand2025marginal,kamin2025renyi,navarro2026finite}.
For this work, we focus on these types of entropies.

However, in various security proofs, there arise situations where one does not analyze exactly some well-characterized state, but instead a different state ``close'' to it under some measure~\cite{kamin2025renyi,wiesemann2026security,upadhyaya2022improved,kanitschar2023finite,navarro2026effective,sekatski2021device}. For instance, this occurs~\cite{kamin2025renyi,wiesemann2026security} in the study of devices with small imperfections, so they do not produce the ideal desired state but rather one close to it. It also arises in ``dimension-reduction'' techniques~\cite{upadhyaya2021dimension,upadhyaya2022improved,kanitschar2023finite,navarro2026effective} in QKD, where infinite-dimensional states are tackled by showing they are close to more tractable finite-dimensional ones. In such situations, one commonly used approach is to apply a ``continuity'' argument --- qualitatively, the idea that if two states are close, then their entropies (and thus the resulting keyrates) should also not be too different. This is formalized by  \emph{entropic continuity bounds}, which rigorously bound the difference in entropy between states that are close according to some measure. 
% However, physical QKD implementations face a practical issue that realistic QKD devices are generally ``imperfect'', differing at least slightly from the models used in idealized security proofs. Hence substantial effort has been put into modifying the proofs to accommodate such imperfections. When the imperfections are small, one commonly used approach is to apply a ``continuity'' argument --- qualitatively, the idea that if the imperfect states are close to the idealized ones, then their entropies (and thus the resulting keyrates) should also not be too different. This is formally captured by  \emph{entropic continuity bounds}, which rigorously bound the difference in entropy between states that are ``close" according to some measure. 

There now exists a rather wide variety of entropic continuity bounds, including but not limited to~\cite{audenaert2007sharp,winter2016tight,marwah2022uniform,bluhm2026unified,cheng2026sharp,berta2026sharp,sekatski2021device,liu2023lipschitz}. However, a subtle limitation of most such existing results (e.g.~\cite{audenaert2007sharp,winter2016tight,marwah2022uniform,bluhm2026unified,cheng2026sharp,berta2026sharp}) is that they use trace distance~\cite{NC10} as the measure of closeness between the states. Under this choice, there exist explicit examples~\cite{bluhm2026unified,cheng2026sharp} demonstrating that the resulting continuity bound 
%must grow faster than any linear function of the trace distance (formally: it cannot be a Lipshitz continuity bound). 
grows very rapidly with trace distance.
This has unfortunate consequences for their applications in QKD security proofs, as a small imperfection in terms of trace distance results in a large change in entropy.

To circumvent this issue, \cite{sekatski2021device} proposed entropic continuity bounds based on fidelity~\cite{NC10}, another measure of the closeness between quantum states. We highlight that this change in the distinguishability measure is not merely an arbitrary or unmotivated choice --- in many QKD applications~\cite{kamin2025renyi,wiesemann2026security,upadhyaya2021dimension,upadhyaya2022improved,kanitschar2023finite,navarro2026effective,sekatski2021device} for instance, the more ``natural'' measure that arises is in fact basically fidelity rather than trace distance. From this perspective, one could say that results based on trace-distance entropic continuity bounds were therefore unnecessarily pessimistic, as they were using a suboptimal choice of closeness measure. It was indeed found in \cite{sekatski2021device} that for entropic continuity bounds based on fidelity, much better results could be obtained.
Unfortunately, not many fidelity-based entropic continuity bounds are currently known, and existing ones have the drawback of being limited to states with classical conditioning registers (making them less suitable for QKD)~\cite{liu2023lipschitz}, having dependencies on the dimension of the conditioning system~\cite{sekatski2021device} (making them unusable when the dimension of that system is unknown), or only applying to specific structured states~\cite{holevo2026log}.

Our contribution in this work is a variety of entropic continuity bounds based on fidelity, covering the von Neumann entropy $H(A|B)$ and the {\Renyi} entropy $\Hup(A|B)$ (more precisely, the ``uparrow sandwiched'' family of {\Renyi} entropies; we provide a detailed definition in Definition~\ref{def:Hup}). For ease of reference, our main result is summarized as follows, in terms of angular distance (a well-known metric based on fidelity~\cite{NC10}; see Definition~\ref{def:angle} for a self-contained presentation).

Note that for convenience of the proofs, all entropies in this work are defined with respect to the natural logarithm (i.e.~in base $e$) --- to obtain results for base-$2$ entropies, simply divide both sides of the expression by $\ln 2$.  
\begin{theorem}\label{thm:intro}
Consider any normalized quantum states $\rho_{AB},\sigma_{AB}$, and any $\alpha\in(\frac12,1)\cup(1,\infty)$. Let $\Lambda(\rho,\sigma)$ denote their angular distance (Definition~\ref{def:angle}). Then for $q=\frac{\alpha-1}{\alpha}$, we have (for entropies \emph{in base $e$}):
\begin{align}
\Big|\Hup(A|B)_\rho-\Hup(A|B)_\sigma\Big| \;\le\; 2\,\Lambda(\rho_{AB},\sigma_{AB})\,\frac{\sinh\big(|q|\arccosh (d_A^\chi)\big)}{|q|},
\end{align}
where $\chi=1/2$ if $\rho,\sigma$ are both classical on $A$, and $\chi=1$ otherwise. 
When $\alpha\to1$, this reduces to
\begin{align}
\big|H(A|B)_\rho-H(A|B)_\sigma\big| \;\le\; 2\,\Lambda(\rho_{AB},\sigma_{AB})\,\ln\!\Big(d_A^\chi+\sqrt{d_A^{2\chi}-1}\Big). \label{eq:introvNbound}
\end{align}
\end{theorem}

In addition to this, we note that in some QKD applications, the quantity of interest is not strictly the {\Renyi} entropy itself, but rather a particular exponentiation of it. For the purposes of such applications, we present in Appendix~\ref{app:quasient} two variants (Corollaries~\ref{cor:Rstate}--\ref{cor:Radd}) of our above continuity bound, which can potentially give better results in such cases. Readers interested in QKD applications should study the results presented in that appendix. 

\begin{remark}
Throughout this work, for ease of analysis we take all systems to be finite-dimensional by default. However, for applications we highlight that our final continuity results (Theorem~\ref{thm:intro}, Proposition~\ref{prop:Rarcsin}, Corollaries~\ref{cor:Rstate}--\ref{cor:Radd}) also automatically hold if the conditioning system $B$ has countably infinite dimension (i.e.~is on a separable Hilbert space), via an essentially straightforward limiting argument we present in Appendix~\ref{app:limits}.

Similarly, while by default we mostly exclude the ``endpoints'' $\alpha \in \{1/2, \infty\}$ for ease of analysis, the final results in (Theorem~\ref{thm:intro}, Proposition~\ref{prop:Rarcsin}, Corollaries~\ref{cor:Rstate}--\ref{cor:Radd}) can be easily extended to those $\alpha$ values by taking the respective limits, as we elaborate on in Appendix~\ref{app:limits}. 
In the case of $\alpha=1/2$, this limit simply consists of literally substituting that value into the expressions. 
\end{remark}

Regarding the bound~\eqref{eq:introvNbound} for the special case of $\alpha=1$, we also performed some preliminary comparisons to the bound obtained in~\cite{liu2023lipschitz}, under the condition that the states are classical on \emph{both} registers $AB$, so the results of both works are applicable. 
% We summarize our results in Sec.~\ref{sec:compare_shannon}.
We find that for $d_A = 2$ or $d_A=3$, our bound~\eqref{eq:introvNbound} is slightly worse than the bound in~\cite{liu2023lipschitz}, suggesting perhaps some potential room for sharpening (albeit with the understanding that our result allows $B$ to be quantum while that work instead allows $A$ to be quantum, so the bounds are not quite directly comparable). However, the discrepancy is fairly small, so the prospective gains may be limited. For $d_A \geq 4$, our bound provides tighter results.

The rest of the paper is structured as follows. In Sec.~\ref{sec:prelim} we lay out our definitions and notation. In Sec.~\ref{sec:mainproofs} we present our proofs of our main results. 
% In Sec.~\ref{sec:compare_shannon} we compare our findings for $\alpha=1$ to~\cite{liu2023lipschitz}. 
Some technical points are deferred to the appendices, though we highlight Appendix~\ref{app:quasient} in particular for results specifically relevant to QKD.

\section{Preliminaries and Notation}
\label{sec:prelim}

Throughout we work over finite dimensional Hilbert spaces (unless otherwise stated) labelled with capital letters $A,B,C \cdots$ with dimension $d_A,d_B,d_C \cdots$. Bounded operators acting on a Hilbert space $A$ are denoted by $B(A)$. Positive semi-definite operators are denoted by $X \geq 0$, and positive definite operators denoted by $X >0$.  

Quantum states on $A$ (positive semi-definite operators of unit trace) are denoted by $S(A)$. Bipartite states are states on the tensor product space $S(AB)$. For a state $\rho_{AB}$, its marginals are denoted with the convention that $\rho_A\coloneqq\Tr_B[\rho_{AB}]$. We suppose that each Hilbert space is inherently equipped with a privileged orthonormal basis called its classical basis; when we say that some state(s) are classical on the corresponding register, such statements should be understood as always being with respect to that basis.
Quantum channels are completely positive trace-preserving maps between the space of bounded operators. 

Inverses and negative powers, when written for potentially singular operators, are defined on their support. $\norm{\cdot}_{p}$ refers to the Schatten $p$-norms, and when subscripts are omitted, the inner products and norms are always the one defined with respect to the chosen basis. For a matrix valued function $X:[a,b] \subset \mathbb{R} \to \mathbb{C}^{n \times n}$, we denote its derivative by $\dot{X}$. 

Throughout this work, we use natural logarithms to define divergences, entropies, and similar such quantities. 

\begin{definition}\label{def:Hup}
Let $\alpha\in(\frac12,1)\cup(1,\infty)$. For a state $\rho$ and $\sigma\ge0$, the sandwiched R\'enyi divergence is
\begin{align}
\Dsand(\rho\|\sigma)\coloneqq\frac{1}{\alpha-1}\ln\Tra\Big[\Big(\sigma^{\frac{1-\alpha}{2\alpha}}\rho\,\sigma^{\frac{1-\alpha}{2\alpha}}\Big)^{\alpha}\Big],
\end{align}
with $\Dsand(\rho\|\sigma)\coloneqq+\infty$ if $\alpha>1$ and $\mathrm{supp}(\rho)\not\subseteq\mathrm{supp}(\sigma)$. For $\rho_{AB}\in S(AB)$, the (uparrow) sandwiched R\'enyi conditional entropy is
\begin{align}
\Hup(A|B)_\rho\coloneqq-\min_{\sigma_B\in S(B)}\Dsand\big(\rho_{AB}\,\big\|\,\id_A\otimes\sigma_B\big).
\end{align}
Similarly, for $\alpha\in(0,1)\cup(1,\infty)$, the Petz R\'enyi divergence is
\begin{align}
\overline{D}_\alpha(\rho\|\sigma)\coloneqq\frac{1}{\alpha-1}\ln\Tra\big[\rho^{\alpha}\sigma^{1-\alpha}\big],
\end{align}
with $\overline{D}_\alpha(\rho\|\sigma)\coloneqq+\infty$ if $\alpha>1$ and $\mathrm{supp}(\rho)\not\subseteq\mathrm{supp}(\sigma)$, and for $\rho_{AB}\in S(AB)$ the (downarrow) Petz R\'enyi conditional entropy is
\begin{align}
\overline{H}^{\downarrow}_\alpha(A|B)_\rho\coloneqq-\overline{D}_\alpha\big(\rho_{AB}\,\big\|\,\id_A\otimes\rho_B\big).
\end{align}
\end{definition}
\medskip
We introduce also the following notation to simplify certain computations later. 
\noindent For $\Sig>0$ and a path of states $\omega_s$ (parametrized by a real parameter $s$), we write
\begin{align}
\gamma\coloneqq\frac{1-\alpha}{2\alpha},\qquad X_s\coloneqq\Sig^{\gamma}\omega_s\Sig^{\gamma},\qquad Q_s\coloneqq\Tra\big[X_s^{\alpha}\big],
\end{align}
so that $\Dsand(\omega_s\|\Sig) = \frac{1}{\alpha-1}\ln Q_s$.

\noindent To establish the bound for the uparrow Sandwiched Renyi entropies, we require the following characterization of its optimizer.

\begin{lemma}[{\cite[Lemma~5]{HT16}}]\label{lem:opt}
    Let $\rho_{AB} \in S(AB)$ and $\tau_A\ge0$ such that $\tau_A\gg\rho_A$, and let $\alpha>\frac{1}{2}$. Then there exists a unique state $\sigmaopt\gg\rho_B$ such that
    \begin{align}
    \sigmaopt
    = \frac{\Tra_A\!\left[\Big(\big(\tau_A\otimes\sigmaopt\big)^{\frac{1-\alpha}{2\alpha}}\,\rho_{AB}\,\big(\tau_A\otimes\sigmaopt\big)^{\frac{1-\alpha}{2\alpha}}\Big)^{\alpha}\right]}
           {\Tra\!\left[\Big(\big(\tau_A\otimes\sigmaopt\big)^{\frac{1-\alpha}{2\alpha}}\,\rho_{AB}\,\big(\tau_A\otimes\sigmaopt\big)^{\frac{1-\alpha}{2\alpha}}\Big)^{\alpha}\right]},
    \label{eq:fixedpoint}
    \end{align}
    which satisfies
    \begin{align}
    \widetilde{I}_\alpha(\rho_{AB}\|\tau_A) \;=\; \Dsand\big(\rho_{AB}\,\big\|\,\tau_A\otimes\sigmaopt\big).
    \end{align}
    For $\alpha=1$ the fixed point of this map is given uniquely by $\Tilde{\sigma}^{*}_{B}(1)=\rho_B$.
\end{lemma}

\noindent Here $\tau\gg\rho$ means $\mathrm{supp}(\tau)\supseteq\mathrm{supp}(\rho)$, and $\widetilde{I}_\alpha(\rho_{AB}\|\tau_A)\coloneqq\min_{\sigma_B\in S(B)}\Dsand(\rho_{AB}\|\tau_A\otimes\sigma_B)$, where for $\alpha=1$ we read $\Dsand$ as the relative entropy. We only use the case $\tau_A=\id_A$, for which $\widetilde{I}_\alpha(\rho_{AB}\|\id_A)=-\Hup(A|B)_\rho$, so that
\begin{align}
\Hup(A|B)_\rho \;=\; -\Dsand\big(\rho_{AB}\,\big\|\,\id_A\otimes\sigmaopt\big), \label{eq:Hupopt}
\end{align}
and~\eqref{eq:fixedpoint} holds with $\tau_A=\id_A$.

\begin{definition}\label{def:angle}
For $\rho,\sigma\in S(A)$, let $F(\rho,\sigma)\coloneqq\big\|\sqrt{\rho}\sqrt{\sigma}\big\|_1$ denote the fidelity, and define
\begin{align}
\Lambda(\rho,\sigma)\coloneqq\arccos F(\rho,\sigma)\in\big[0,\tfrac{\pi}{2}\big].
\end{align}
Note that $\Lambda$ defines a metric on quantum states~\cite{NC10}; this fact is later used in proving our bounds when the $B$ system is of countably infinite dimension (after a suitable generalization to that scenario).

\end{definition}
\section{Entropic continuity bound}
\label{sec:mainproofs}
We first describe the overall strategy. For simplicity, we restrict the discussion to the von Neumann conditional entropy $H(A|B)$.

Let $\rho_{AB},\sigma_{AB}\in S(AB)$, and let $s\mapsto\omega_s$, $s\in[0,\theta]$, be a smooth path of full-rank states with $\omega_0=\rho_{AB}$ and $\omega_\theta=\sigma_{AB}$. By the fundamental theorem of calculus,
\begin{align}
\big|H(A|B)_\rho-H(A|B)_\sigma\big|\;\le\;\int_0^\theta\Big|\frac{d}{ds}H(A|B)_{\omega_s}\Big|\,ds,
\qquad
\frac{d}{ds}H(A|B)_{\omega_s}=-\Tra\big[\dot\omega_s L_s\big],
\end{align}
where $L_s\coloneqq\ln\omega_s-\ln(\id_A\otimes\omega_{B,s})$. Writing $\omega_s$ as the marginal of a smooth path of unit vectors $|\Psi_s\rangle$, the Cauchy--Schwarz inequality bounds this derivative by $2\|\dot\Psi_s\|\sqrt{V(A|B)_{\omega_s}}$.

A dimension-dependent bound on $V(A|B)$ therefore yields a continuity bound in terms of the length $\int_0^\theta\|\dot\Psi_s\|\,ds$ of the path. To make this length as small as possible, it is natural, in the spirit of~\cite{liu2023lipschitz}, to take the reduced great-circle path between Uhlmann purifications of $\rho_{AB}$ and $\sigma_{AB}$, whose length is $\Lambda(\rho_{AB},\sigma_{AB})$.

Dimension-dependent bounds on the entropy variance were previously obtained in~\cite{rw15} for the unconditional entropy, and in~\cite{FD18} for the conditional entropy. Our strategy is to generalise this approach to $\Hup$. We handle the optimisation over $\sigma_B$ in its definition with Danskin's theorem, as done in~\cite{cheng2026sharp}, so that the entropy variance is replaced by the $\alpha$-variance $V_\alpha$ at the optimizer (Definition~\ref{def:Valpha}), after which it remains to bound the variance term.

\begin{restatable}{lemma}{pathlemma}\label{lem:path}
Let $\rho_{AB},\sigma_{AB},\pi_{AB}\in S(AB)$ and $\theta\coloneqq\Lambda(\rho_{AB},\sigma_{AB})>0$. Let $\pi_{AB}>0$ and $\varepsilon\in(0,1)$. Then there exist a smooth path $s\mapsto\omega_s$ in $S(AB)$, $s\in[0,\theta]$, with $\omega_0=\rho_{AB}$ and $\omega_\theta=\sigma_{AB}$ (the reduced great-circle path between Uhlmann purifications of $\rho_{AB}$ and $\sigma_{AB}$), a purifying system $C$, and a smooth path of unit vectors $s\mapsto|\Psi_s^\varepsilon\rangle_{ABC}$ such that
\begin{align}
\Tra_C\big[|\Psi_s^\varepsilon\rangle\!\langle\Psi_s^\varepsilon|\big] &= (1-\varepsilon)\,\omega_s + \varepsilon\,\pi_{AB}, \\
\|\dot\Psi_s^\varepsilon\| &= \sqrt{1-\varepsilon}, \label{eq:pathspeed}\\
\langle\Psi_s^\varepsilon|\dot\Psi_s^\varepsilon\rangle &= 0.
\end{align}
\end{restatable}
\noindent The proof is given in Appendix~\ref{app:proofs}.

To apply the same argument to classical-quantum states, we need to upper-bound the (R\'enyi) entropy variance at \textit{some} point of the path from $\rho_{AB}$ to $\sigma_{AB}$. Since that point is not known in advance, the bound must hold uniformly along the path. If the state is classical on $A$, Proposition~\ref{prop:Vbounds} gives a strictly better bound on the variance (for $d_A\ge2$), so we obtain a strictly better continuity bound by applying the mean-value theorem along a path on which every state is classical on $A$. We therefore require the lemma below, which takes the path of Lemma~\ref{lem:path} and pinches it on $A$ in the classical basis common to $\rho_{AB}$ and $\sigma_{AB}$, while keeping a purification with the same properties.
\begin{restatable}{lemma}{pathlemmacq}\label{lem:pathcq}
Let
% $\{|a\rangle\}_a$  be the classical basis of $\cH_A$, and let 
$\rho_{AB},\sigma_{AB},\pi_{AB}\in S(AB)$ be classical on $A$. Let $\theta\coloneqq\Lambda(\rho_{AB},\sigma_{AB})>0$, $\pi_{AB}>0$ and $\varepsilon\in(0,1)$. Then there exist a smooth path $s\mapsto\widetilde\omega_s$ in $S(AB)$, $s\in[0,\theta]$, with $\widetilde\omega_0=\rho_{AB}$, $\widetilde\omega_\theta=\sigma_{AB}$ and $\widetilde\omega_s$ classical on $A$ 
% with respect to $\{|a\rangle\}_a$ 
for every $s$, a purifying system $C$, and a smooth path of unit vectors $s\mapsto|\widetilde\Psi_s^\varepsilon\rangle_{ABC}$ such that
\begin{align}
\Tra_C\big[|\widetilde\Psi_s^\varepsilon\rangle\!\langle\widetilde\Psi_s^\varepsilon|\big] &= (1-\varepsilon)\,\widetilde\omega_s + \varepsilon\,\pi_{AB}, \\
\|\dot{\widetilde\Psi}{}_s^\varepsilon\| &= \sqrt{1-\varepsilon}, \\
\langle\widetilde\Psi_s^\varepsilon|\dot{\widetilde\Psi}{}_s^\varepsilon\rangle &= 0.
\end{align}
\end{restatable}
\noindent The proof is given in Appendix~\ref{app:proofs}.

\begin{lemma}\label{lem:Dsmooth}
Let $s\mapsto\omega_s$ be a smooth path of full rank states. Then, for every fixed $\Sig>0$, the function $s\mapsto\Dsand(\omega_s\|\Sig)$ is differentiable with
\begin{align}
\frac{d}{ds}\Dsand(\omega_s\|\Sig) = \frac{\alpha}{\alpha-1}\,\frac{\Tra\big[X_s^{\alpha-1}\Sig^{\gamma}\dot\omega_s\Sig^{\gamma}\big]}{Q_s}. \label{eq:Dsderiv}
\end{align}
Moreover, both $\Dsand(\omega_s\|\Sig)$ and the right-hand side of~\eqref{eq:Dsderiv} are jointly continuous in $(s,\Sig)$ on the set of $\Sig>0$.
\end{lemma}
\begin{proof}
We have
\begin{align}
\frac{d}{ds}\Dsand(\omega_s\|\Sig) &= \frac{1}{\alpha-1}\frac{d}{ds}\ln(Q_s) \notag\\
&= \frac{1}{\alpha-1}\Big(\frac{\dot Q_s}{Q_s}\Big).
\end{align}
Since $s\mapsto X_s$ is smooth with $X_s>0$, we have
\begin{align}
\frac{d}{ds}\Tra\big[X_s^\alpha\big] = \alpha\,\Tra\big[X_s^{\alpha-1}\dot X_s\big], \label{eq:tracederiv}
\end{align}
 and $\dot X_s=\Sig^{\gamma}\dot\omega_s\Sig^{\gamma}$, so
\begin{align}
\dot Q_s = \alpha\,\Tra\big[X_s^{\alpha-1}\Sig^{\gamma}\dot\omega_s\Sig^{\gamma}\big],
\end{align}
which gives~\eqref{eq:Dsderiv}. For the continuity, note that $\Sig\mapsto\Sig^{\gamma}$, $X\mapsto X^{\alpha-1}$ and $X\mapsto X^{\alpha}$ are continuous on positive definite operators by the continuous functional calculus. Since $X_s=\Sig^{\gamma}\omega_s\Sig^{\gamma}$ is positive definite and continuous in $(s,\Sig)$, so are $Q_s>0$ and $X_s^{\alpha-1}$, and $\dot\omega_s$ is continuous in $s$.
\end{proof}

\begin{definition}[$\alpha$-variance]\label{def:Valpha}
Let $\alpha\in(\frac12,1)\cup(1,\infty)$, $\gamma=\frac{1-\alpha}{2\alpha}$ as above, and $q\coloneqq\frac{\alpha-1}{\alpha}$. For a state $\omega$ and $\Sig>0$, let
\begin{align}
X(\omega\|\Sig) \coloneqq \Sig^{\gamma}\omega\Sig^{\gamma}, \qquad
Q(\omega\|\Sig) \coloneqq \Tra\big[X(\omega\|\Sig)^{\alpha}\big], \qquad
Y_\alpha(\omega\|\Sig) \coloneqq \frac{\Sig^{\gamma}X(\omega\|\Sig)^{\alpha-1}\Sig^{\gamma}}{Q(\omega\|\Sig)}.
\end{align}
The $\alpha$-variance of $\omega$ relative to $\Sig$ is
\begin{align}
V_\alpha(\omega\|\Sig) \coloneqq \frac{\Tra\big[\omega Y_\alpha(\omega\|\Sig)^2\big]-1}{q^2}. \label{eq:Valpha}
\end{align}
\end{definition}
\noindent In this notation, $X_s=X(\omega_s\|\Sig)$ and $Q_s=Q(\omega_s\|\Sig)$, and by cyclicity of the trace~\eqref{eq:Dsderiv} reads
\begin{align}
\frac{d}{ds}\Dsand(\omega_s\|\Sig) = \frac{\alpha}{\alpha-1}\Tra\big[\dot\omega_s Y_\alpha(\omega_s\|\Sig)\big]. \label{eq:DsderivR}
\end{align}
\noindent $V_\alpha$ is a variance-like quantity: $Y_\alpha(\omega\|\Sig)$ is Hermitian and $\Tra[\omega Y_\alpha(\omega\|\Sig)] = \Tra[X^{\alpha}]/Q = 1$, so
\begin{align}
V_\alpha(\omega\|\Sig) = \Tra\Big[\omega\Big(\frac{Y_\alpha}{q}\Big)^2\Big] - \Tra\Big[\omega\,\frac{Y_\alpha}{q}\Big]^2 = \frac{\Tra\big[\omega(Y_\alpha-\id)^2\big]}{q^2} \;\ge\; 0, \label{eq:Vvariance}
\end{align}
i.e.\ $V_\alpha(\omega\|\Sig)$ is the variance of the observable $Y_\alpha/q$ in the state $\omega$. In the limit $\alpha\to1$ it recovers the usual relative entropy variance. For a state $\omega$ and $\Sig>0$, write (with $\ln$ the natural logarithm)
\begin{align}
D(\omega\|\Sig)\coloneqq\Tra\big[\omega(\ln\omega-\ln\Sig)\big], \qquad
V(\omega\|\Sig)\coloneqq\Tra\big[\omega(\ln\omega-\ln\Sig)^2\big]-D(\omega\|\Sig)^2,
\end{align}
and $V(A|B)_\rho\coloneqq V(\rho_{AB}\|\id_A\otimes\rho_B)$.

\begin{proposition}\label{prop:Vlimit}
Let $\omega>0$ be a state.
\begin{enumerate}[label=(\roman*)]
\item For every fixed $\Sig>0$, $\displaystyle\lim_{\alpha\to1}V_\alpha(\omega\|\Sig) = V(\omega\|\Sig)$.
\item For $\rho_{AB}>0$ and $\Sig_\alpha\coloneqq\id_A\otimes\sigmaopt$, with $\sigmaopt$ the optimizer of Lemma~\ref{lem:opt}, $\displaystyle\lim_{\alpha\to1}V_\alpha(\rho_{AB}\|\Sig_\alpha) = V(A|B)_\rho$.
\end{enumerate}
\end{proposition}
\begin{proof}
(i) Write $L\coloneqq\ln\omega-\ln\Sig$ and $D\coloneqq D(\omega\|\Sig)=\Tra[\omega L]$. As $\alpha\to1$ we have $q\to0$, $\gamma=-q/2$ and $\alpha-1=q\alpha=q+O(q^2)$. Since $\omega,\Sig>0$,
\begin{align}
\Sig^{\gamma} = \id-\tfrac{q}{2}\ln\Sig+O(q^2), \qquad
X^{\alpha-1} = \id+q\ln\omega+O(q^2),
\end{align}
where for the second expansion we used $X=\Sig^\gamma\omega\Sig^\gamma=\omega+O(q)$. Hence
\begin{align}
\Sig^{\gamma}X^{\alpha-1}\Sig^{\gamma} = \id+qL+O(q^2), \qquad
Q=\Tra\big[\omega\,\Sig^{\gamma}X^{\alpha-1}\Sig^{\gamma}\big] = 1+qD+O(q^2),
\end{align}
so that $Y_\alpha = \id+q(L-D\,\id)+O(q^2)$. By~\eqref{eq:Vvariance},
\begin{align}
V_\alpha(\omega\|\Sig) = \frac{\Tra\big[\omega(Y_\alpha-\id)^2\big]}{q^2} = \Tra\big[\omega(L-D\,\id)^2\big]+O(q) \;\longrightarrow\; \Tra[\omega L^2]-D^2 = V(\omega\|\Sig).
\end{align}

(ii) The expansion in (i) is uniform for $\Sig$ in a compact set of strictly positive operators, so it suffices that $\Sig_\alpha\to\id_A\otimes\rho_B$ as $\alpha\to1$. By Lemma~\ref{lem:opt}, $\Tilde{\sigma}^{*}_{B}(1)=\rho_B$, and $\alpha\mapsto\sigmaopt$ is continuous at $\alpha=1$~\cite[Appendix~D]{HT16}. 
\end{proof}

\begin{proposition}\label{prop:Dderiv}
For every fixed $\Sig>0$, every smooth full rank path of states $\omega_s$, and every smooth path of unit vectors $|\Psi_s\rangle_{ABC}$ with $\Tra_C\big[|\Psi_s\rangle\!\langle\Psi_s|\big]=\omega_s$, we have
\begin{align}
\Big|\frac{d}{ds}\Dsand(\omega_s\|\Sig)\Big| \;\le\; 2\,\|\dot\Psi_s\|\sqrt{V_\alpha(\omega_s\|\Sig)}.
\end{align}
\end{proposition}
\begin{proof}
By Lemma~\ref{lem:Dsmooth} in the form~\eqref{eq:DsderivR},
\begin{align}
\frac{d}{ds}\Dsand(\omega_s\|\Sig) = \frac{\alpha}{\alpha-1}\Tra\big[\dot\omega_s Y_\alpha(\omega_s\|\Sig)\big].
\end{align}
Let
\begin{align}
A \coloneqq \frac{\alpha}{\alpha-1}Y_\alpha.
\end{align}
Since $\omega_s = \Tra_C\big[|\Psi_s\rangle\!\langle\Psi_s|\big]$,
\begin{align}
\dot\omega_s = \Tra_C\big[|\dot\Psi_s\rangle\!\langle\Psi_s| + |\Psi_s\rangle\!\langle\dot\Psi_s|\big].
\end{align}
So a straightforward calculation gives
\begin{align}
\frac{\alpha}{\alpha-1}\Tra\big[\dot\omega_s Y_\alpha(\omega_s\|\Sig)\big] = 2\,\Re\langle\dot\Psi_s|A\otimes\id_C|\Psi_s\rangle,
\end{align}
Since $\alpha/(\alpha-1) = 1/q$, we have $A = Y_\alpha/q$. As $\|\Psi_s\|=1$, $\Re\langle\dot\Psi_s|\Psi_s\rangle = \frac12\frac{d}{ds}\|\Psi_s\|^2 = 0$, so we may centre $A$:
\begin{align}
2\,\Re\langle\dot\Psi_s|A\otimes\id_C|\Psi_s\rangle = 2\,\Re\big\langle\dot\Psi_s\big|\big(A-\tfrac{1}{q}\id\big)\otimes\id_C\big|\Psi_s\big\rangle.
\end{align}
By Cauchy--Schwarz,
\begin{align}
\Big|\frac{d}{ds}\Dsand(\omega_s\|\Sig)\Big| \;\le\; 2\,\|\dot\Psi_s\|\,\big\|\big((A-\tfrac1q\id)\otimes\id_C\big)|\Psi_s\rangle\big\|,
\end{align}
and, using $\Tra[\omega_s Y_\alpha]=1$,
\begin{align}
\big\|\big((A-\tfrac1q\id)\otimes\id_C\big)|\Psi_s\rangle\big\|^2
&= \Tra\big[\omega_s\big(A-\tfrac1q\id\big)^2\big]
= \frac{1}{q^2}\Tra\big[\omega_s(Y_\alpha-\id)^2\big] \notag\\
&= \frac{\Tra[\omega_s Y_\alpha^2] - 2\Tra[\omega_s Y_\alpha] + 1}{q^2}
= V_\alpha(\omega_s\|\Sig). \qedhere
\end{align}
\end{proof}

\begin{lemma}[Danskin~\cite{danskin1966theory}, Lemma 4.1.1 of \cite{cardaliaguet2021shortcourse}]\label{lem:danskin}
Let $K$ be a compact metric space, $I\subseteq\mathbb{R}$ an open interval, and $f:I\times K\to\mathbb{R}$ continuous, with $\partial_s f$ existing and continuous on $I\times K$. Let $\phi(s)\coloneqq\min_{x\in K}f(s,x)$. If for every $s\in I$ the minimiser $x_s$ is unique, then $\phi$ is differentiable on $I$ and
\begin{align}
\phi'(s) = \partial_s f(s,x)\big|_{x=x_s}.
\end{align}
\end{lemma}

\begin{lemma}\label{lem:Hderiv}
Let $\rho_{AB},\sigma_{AB}\in S(AB)$ with $\theta\coloneqq\Lambda(\rho_{AB},\sigma_{AB})>0$, let $\pi_{AB}>0$, $\varepsilon\in(0,1)$ and $\alpha\in(\frac12,1)\cup(1,\infty)$. Let $\omega_s$ and $|\Psi^\varepsilon_s\rangle$ be as in Lemma~\ref{lem:path}, and set
\begin{align}
\omega^\varepsilon_s\coloneqq(1-\varepsilon)\,\omega_s+\varepsilon\,\pi_{AB}, \qquad \Sig^\varepsilon_s\coloneqq\id_A\otimes\sigma^\varepsilon_s,
\end{align}
where $\sigma^\varepsilon_s$ denotes the optimizer of Lemma~\ref{lem:opt} for the state $\omega^\varepsilon_s$. Then
\begin{align}
h_\varepsilon(s)\coloneqq\Hup(A|B)_{\omega^\varepsilon_s}
\end{align}
is continuous on $[0,\theta]$, differentiable on $(0,\theta)$, and
\begin{align}
|h_\varepsilon'(s)| \;\le\; 2\sqrt{1-\varepsilon}\,\sqrt{V_\alpha(\omega^\varepsilon_s\|\Sig^\varepsilon_s)} \qquad\text{for all } s\in(0,\theta). \label{eq:Hderiv}
\end{align}
\end{lemma}
\begin{proof}
The idea is to apply Danskin's theorem (Lemma~\ref{lem:danskin}). To do so, we need to choose a compact set $K\subseteq S(B)$ such that
\begin{enumerate}[label=\arabic*.]
\item $\sigma^\varepsilon_s\in K$ for all $s\in[0,\theta]$;
\item $\Dsand(\omega^\varepsilon_s\|\id_A\otimes\tau_B)<\infty$ for all $\tau_B\in K$ and $s\in[0,\theta]$.
\end{enumerate}
We first consider the case $\alpha>1$. From Lemma~\ref{lem:opt} we know that $\sigma^\varepsilon_s$ satisfies the fixed-point equation
\begin{align}
\sigma^\varepsilon_s=\frac{\Tra_A\big[X_s^{\alpha}\big]}{Q_s},
\qquad X_s\coloneqq(\Sig^\varepsilon_s)^{\gamma}\,\omega^\varepsilon_s\,(\Sig^\varepsilon_s)^{\gamma},
\qquad Q_s\coloneqq\Tra\big[X_s^{\alpha}\big].
\end{align}
So it suffices to pick
\begin{align}
\delta\coloneqq\big(\varepsilon\lambda_{\min}(\pi_{AB})\big)^{\alpha}d_A^{2-\alpha},
\qquad\text{and set}\qquad
K\coloneqq\big\{\tau_B\in S(B):\tau_B\ge\delta\,\id_B\big\}.
\end{align}
Now set
\begin{align}
f_\varepsilon(s,\tau_B)\coloneqq\Dsand\big(\omega^\varepsilon_s\,\big\|\,\id_A\otimes\tau_B\big), \label{eq:feps}
\end{align}
so that $h_\varepsilon(s)=-\min_{\tau_B\in S(B)}f_\varepsilon(s,\tau_B)$. Then, by Lemma~\ref{lem:Dsmooth} and~\eqref{eq:DsderivR}, $f_\varepsilon$ is continuous on $[0,\theta]\times K$ and differentiable in $s$ with
\begin{align}
\partial_sf_\varepsilon(s,\tau_B)=\frac{\alpha}{\alpha-1}\,\Tra\big[\dot\omega^\varepsilon_s\,Y_\alpha(\omega^\varepsilon_s\|\id_A\otimes\tau_B)\big],
\qquad \dot\omega^\varepsilon_s=(1-\varepsilon)\,\dot\omega_s. \label{eq:dsfeps}
\end{align}
Again by Lemma~\ref{lem:Dsmooth}, $\partial_sf_\varepsilon$ is continuous on $[0,\theta]\times K$.

By Lemma~\ref{lem:opt}, $\sigma^\varepsilon_s$ is the unique minimiser of $f_\varepsilon(s,\cdot)$ over $S(B)$. Since it lies in $K$, it is also the unique minimiser over $K$, and $h_\varepsilon(s)=-\min_{\tau_B\in K}f_\varepsilon(s,\tau_B)$. Lemma~\ref{lem:danskin}, applied to $f_\varepsilon$ on $(0,\theta)\times K$, therefore shows that $h_\varepsilon$ is differentiable on $(0,\theta)$ with
\begin{align}
h_\varepsilon'(s)=-\,\partial_sf_\varepsilon(s,\tau_B)\Big|_{\tau_B=\sigma^\varepsilon_s},
\end{align}
and $h_\varepsilon$ is continuous on $[0,\theta]$, being minus the minimum over the compact set $K$ of a continuous function. The right-hand side above is the derivative at $s'=s$ of $s'\mapsto-\Dsand(\omega^\varepsilon_{s'}\|\Sig^\varepsilon_s)$ with $\Sig^\varepsilon_s>0$ held fixed. By Lemma~\ref{lem:path}, $|\Psi^\varepsilon_{s'}\rangle$ is a smooth purification of the full-rank path $\omega^\varepsilon_{s'}$ with $\|\dot\Psi^\varepsilon_{s'}\|=\sqrt{1-\varepsilon}$, see~\eqref{eq:pathspeed}. Proposition~\ref{prop:Dderiv} therefore gives~\eqref{eq:Hderiv}.

For $\alpha\in(\frac12,1)$, we pick $\delta\coloneqq\varepsilon\lambda_{\min}(\pi_{AB})\,d_A$ and the proof follows analogously.
\end{proof}

\bigskip

\bigskip
\noindent\textbf{Notation.} Let $\rho_{AB}\in S(AB)$, $\alpha>\frac12$, $\alpha\neq 1$, and let $\sigmaopt$ be such that
\begin{align}
\Hup(A|B)_\rho = -\Dsand\big(\rho_{AB}\,\big\|\,\id_A\otimes\sigmaopt\big),
\end{align}
as in Lemma~\ref{lem:opt}. Introduce $\Sig = \id_A\otimes\sigmaopt$, and with the notation of Definition~\ref{def:Valpha}, write for shorthand
\begin{align}
X = X(\rho_{AB}\|\Sig), \qquad Q = Q(\rho_{AB}\|\Sig), \qquad Y_\alpha = Y_\alpha(\rho_{AB}\|\Sig), \qquad V_\alpha = V_\alpha(\rho_{AB}\|\Sig).
\end{align}

From the above lemma, what remains to get the desired continuity bound is basically an application of the mean-value theorem. For this, we require upper bounds on the $\alpha$-variance quantity that appeared above. The rough strategy for obtaining these bounds is the following: we first observed that the second moment $\Tr[\rho Y_{\alpha}^2]$ can be written as the product of exponentiated Petz-Rényi divergences between certain states. We then pass to the Nussbaum-Szko{\l}a distributions. Then, a combination of Jensen's and parameterized AM-GM allows us to bound the $\alpha$-variance in terms of the exponentiated Petz-Rényi divergences, for which we have bounds in terms of only the dimension of the $A$-system. 

\noindent We need the following moment inequality.

\begin{restatable}{lemma}{momentlemma}\label{lem:moment}
Let $q\in[-1,1]$ and $f_q(x)=\cosh(q\arccosh x)$ for $x\in[1,\infty)$. Then $f_q$ is increasing and concave on $[1,\infty)$, and for every real random variable $Z$ with finite support,
\begin{align}
\EE\cosh(qZ) \;\le\; f_q\big(\EE\cosh Z\big). \label{eq:moment}
\end{align}
\end{restatable}
\noindent The proof is given in Appendix~\ref{app:proofs}.

\begin{lemma}\label{lem:product}
Let
\begin{align}
M(q) = \Tra\!\left[\Big(\frac{X^\alpha}{Q}\Big)^{1+q}\Sig^{-q}\right] \qquad\text{for } q=\frac{\alpha-1}{\alpha}.
\end{align}
Then $\Tra[\rho Y_\alpha^2] = M(q)\cdot M(-q)$.
\end{lemma}
\begin{proof}
\begin{align}
\Tra[\rho Y_\alpha^2]
&= \frac{1}{Q^2}\Tra\!\left[\rho\big(\Sig^{\gamma}X^{\alpha-1}\Sig^{\gamma}\Sig^{\gamma}X^{\alpha-1}\Sig^{\gamma}\big)\right] \notag\\
&= \frac{1}{Q^2}\Tra\!\left[X\,X^{\alpha-1}\Sig^{2\gamma}X^{\alpha-1}\right] \notag\\
&= \frac{1}{Q^2}\Tra\!\left[X^{2\alpha-1}\Sig^{2\gamma}\right] \notag\\
&= \frac{1}{Q^{1+q}Q^{1-q}}\Tra\!\left[X^{2\alpha-1}\Sig^{2\gamma}\right] \notag\\
&= \frac{1}{Q^{1+q}Q^{1-q}}\Tra\!\left[(X^{\alpha})^{q+1}\Sig^{-q}\right] = M(q)\,\frac{1}{Q^{1-q}}.
\end{align}
Now observing $M(-q) = \frac{1}{Q^{1-q}}$ finishes the proof.
\end{proof}

\begin{lemma}\label{lem:Vgeneral}
Let $\rho_{AB}\in S(AB)$, let $\alpha\in(\frac12,1)\cup(1,\infty)$, $q=\frac{\alpha-1}{\alpha}$, and let $\sigmaopt$ be the optimizer of Lemma~\ref{lem:opt}. Assume $\sigmaopt>0$, write $\Sig\coloneqq\id_A\otimes\sigmaopt$, and let $X\coloneqq X(\rho_{AB}\|\Sig)$ and $Q\coloneqq Q(\rho_{AB}\|\Sig)$ be as in Definition~\ref{def:Valpha}. For $t\in[-1,1]$ define
\begin{align}
M(t)\coloneqq\Tra\!\left[\Big(\frac{X^\alpha}{Q}\Big)^{1+t}\Sig^{-t}\right],
\quad\text{so that}\quad
M(1)=\Tra\!\left[\Big(\frac{X^\alpha}{Q}\Big)^{2}\Sig^{-1}\right],
\quad
M(-1)=\Tra\!\left[\Big(\frac{X^\alpha}{Q}\Big)^{0}\Sig\right],
\end{align}
where $\big(\frac{X^\alpha}{Q}\big)^{0}$ is the projection onto the support of $\frac{X^\alpha}{Q}$. Then
\begin{align}
V_\alpha\big(\rho_{AB}\,\big\|\,\id_A\otimes\sigmaopt\big)\;\le\;\frac{\sinh^2\!\Big(q\arccosh\sqrt{M(1)M(-1)}\Big)}{q^2}. \label{eq:Vgeneral}
\end{align}
\end{lemma}
\begin{proof}
We use the Nussbaum--Szko{\l}a distributions~\cite{NS09}. Let
\begin{align}
\frac{X^\alpha}{Q} = \sum_i \lambda_i\, |u_i\rangle\!\langle u_i| \qquad\text{and}\qquad \Sig = \sum_j \mu_j\, |v_j\rangle\!\langle v_j|
\end{align}
be the spectral decompositions of $\frac{X^\alpha}{Q}$ and $\Sig$ respectively, and define
\begin{align}
P(i,j) = \lambda_i\,|\langle u_i|v_j\rangle|^2, \qquad Z(i,j) = \ln\Big(\frac{\lambda_i}{\mu_j}\Big)
\end{align}
for the pairs $(i,j)$ with $\lambda_i>0$. Since $\Sig>0$ and $\frac{X^\alpha}{Q}$ is a state, $P$ is a probability distribution, and $M(q)=\EE_P[e^{qZ}]$ for $q\in[-1,1]$. In particular, we have $M(1)=\EE_P[e^{Z}]$ and $M(-1)=\EE_P[e^{-Z}]$.
Now fix $\mu\in\mathbb{R}$. Since $e^{-q\mu}M(q)=\EE_P[e^{q(Z-\mu)}]$ and $e^{q\mu}M(-q)=\EE_P[e^{-q(Z-\mu)}]$, Lemma~\ref{lem:product} and the inequality $xy\le\big(\frac{x+y}{2}\big)^2$ give
\begin{align}
\Tra[\rho Y_\alpha^2]
&= \big(e^{-q\mu}M(q)\big)\big(e^{q\mu}M(-q)\big) \notag\\
&\le \Big(\EE_P\big[\cosh\big(q(Z-\mu)\big)\big]\Big)^2 \notag\\
&\le f_q\big(\EE_P[\cosh(Z-\mu)]\big)^2
  && \text{(Lemma~\ref{lem:moment})} \notag\\
&= f_q\Big(\tfrac12\big(e^{-\mu}M(1)+e^{\mu}M(-1)\big)\Big)^2.
\end{align}
Minimizing the right-hand side, we choose $e^{2\mu}=M(1)/M(-1)$. Thus, 
\begin{align}
V_\alpha = \frac{\Tra[\rho Y_\alpha^2]-1}{q^2}
\;\le\; \frac{f_q\big(\sqrt{M(1)M(-1)}\big)^2-1}{q^2}
= \frac{\sinh^2\!\Big(q\arccosh\sqrt{M(1)M(-1)}\Big)}{q^2}. \qedhere
\end{align}
where in the last line we used $f_q(x)^2-1=\sinh^2(q\arccosh x)$.
\end{proof}

\begin{proposition}\label{prop:Vbounds}
Let $\rho_{AB}\in S(AB)$, let $\alpha\in(\frac12,1)\cup(1,\infty)$, $q=\frac{\alpha-1}{\alpha}$, and let $\sigmaopt$ be the optimizer of Lemma~\ref{lem:opt}. Assume $\sigmaopt>0$, which holds for instance if $\rho_B>0$. Then
\begin{align}
V_\alpha\big(\rho_{AB}\,\big\|\,\id_A\otimes\sigmaopt\big)\;\le\;\frac{\sinh^2\!\big(q\arccosh d_A\big)}{q^2}. \label{eq:Vbound}
\end{align}
If moreover $\rho_{AB}$ is classical on $A$, then
\begin{align}
V_\alpha\big(\rho_{AB}\,\big\|\,\id_A\otimes\sigmaopt\big)\;\le\;\frac{\sinh^2\!\big(q\arccosh\sqrt{d_A}\big)}{q^2}. \label{eq:Vboundcq}
\end{align}
\end{proposition}
\begin{proof}
Write $\sigma\coloneqq\sigmaopt$ and $\Sig\coloneqq\id_A\otimes\sigma$, and let $\xi\coloneqq X^\alpha/Q$, with $X$ and $Q$ as in Definition~\ref{def:Valpha} for the pair $(\rho_{AB},\Sig)$. By Lemma~\ref{lem:Vgeneral} it suffices to show that $M(1)M(-1)\le d_A^2$ in general, and that $M(1)M(-1)\le d_A$ if $\rho_{AB}$ is classical on $A$.

We first bound $M(-1)$ in both cases. By our convention for powers, $\xi^0$ is the projection onto the support of $\xi$, so $\xi^0\le\id$, and since $\Sig\ge0$,
\begin{align}
M(-1)=\Tra\big[\xi^0\,\Sig\big]\;\le\;\Tra[\Sig]=\Tra[\id_A\otimes\sigma]=d_A.
\end{align}

Now $\xi$ is a state on $AB$, and by the fixed-point equation~\eqref{eq:fixedpoint} its marginal on $B$ is $\Tra_A[\xi]=\sigma$, i.e.\ $\Sig=\id_A\otimes\Tra_A[\xi]$. Hence
\begin{align}
M(1)=\Tra\big[\xi^2\,\Sig^{-1}\big] &= e^{\overline{D}_2(\xi\|\Sig)} \notag\\
&= e^{-\overline{H}^{\downarrow}_2(A|B)_{\xi}} \notag\\
&\le e^{\ln(d_A)} = d_A,
\end{align}
where the second equality uses that $\Sig$ is $\id_A$ tensored with the $B$-marginal of $\xi$, and the last inequality follows from~\cite{tomamichel2015quantum}. This gives $M(1)M(-1)\le d_A^2$, and hence~\eqref{eq:Vbound}.

If $\rho_{AB}$ is classical on $A$, then so is $\xi$, since $\Sig=\id_A\otimes\sigma$ does not mix the blocks $|a\rangle\!\langle a|_A\otimes\id_B$ (letting $\ket{a}$ denote the classical basis of $A$), and hence neither $X$ nor $X^\alpha$ does. As above,
\begin{align}
M(1)=\Tra\big[\xi^2\,\Sig^{-1}\big] &= e^{\overline{D}_2(\xi\|\Sig)} \notag\\
&= e^{-\overline{H}^{\downarrow}_2(A|B)_{\xi}} \notag\\
&\le e^{0} = 1,
\end{align}
where the inequality holds because the conditional entropy of a state classical on $A$ is non-negative~\cite{tomamichel2015quantum}. Together with $M(-1)\le d_A$, this gives~\eqref{eq:Vboundcq}.
\end{proof}

\begin{corollary}\label{cor:variance}
For any full-rank $\rho_{AB}\in S(AB)$,
\begin{align}
V(A|B)_\rho \;\le\; \ln^2\!\Big(d_A + \sqrt{d_A^2-1}\Big).
\end{align}
\end{corollary}
\begin{proof}
For general $\rho_{AB}$, let $\alpha\to1$ (so $q\to0$) in Proposition~\ref{prop:Vbounds}, using $V_\alpha\to V(A|B)_\rho$ (Proposition~\ref{prop:Vlimit}) and $\sinh(qx)/q\to x$, which gives $V(A|B)_\rho\le\arccosh^2(d_A) = \ln^2\big(d_A+\sqrt{d_A^2-1}\big)$. 

Note that since $V$ and the right-hand side both scale by $1/\ln^2 b$ when changing to base-$b$ logarithms, the bound holds in any base. 
\end{proof}

Since $d_A + \sqrt{d_A^2-1} < 2d_A$, Corollary~\ref{cor:variance} improves on the variance bounds of \cite[Corollary~III.5]{FD18}, which were $V(A|B)_\rho\le\ln^2(2d_A^2+1)$ for general states and $V(A|B)_\rho\le\ln^2(2d_A+1)$ when the state is classical on $A$.

\bigskip
\noindent Combining Lemma~\ref{lem:Hderiv} with Proposition~\ref{prop:Vbounds} gives our continuity bound.

\begin{theorem}\label{thm:continuity}
Let $\rho_{AB},\sigma_{AB}\in S(AB)$, $\alpha\in(\frac12,1)\cup(1,\infty)$ and $q=\frac{\alpha-1}{\alpha}$. Then
\begin{align}
\Big|\Hup(A|B)_\rho-\Hup(A|B)_\sigma\Big| \;\le\; 2\,\Lambda(\rho_{AB},\sigma_{AB})\,\frac{\sinh\big(|q|\arccosh d_A\big)}{|q|}. \label{eq:continuity}
\end{align}
If moreover $\rho_{AB}$ and $\sigma_{AB}$ are classical on $A$,
% with respect to the same orthonormal basis of $\cH_A$, 
then
\begin{align}
\Big|\Hup(A|B)_\rho-\Hup(A|B)_\sigma\Big| \;\le\; 2\,\Lambda(\rho_{AB},\sigma_{AB})\,\frac{\sinh\big(|q|\arccosh \sqrt{d_A}\big)}{|q|}. \label{eq:continuitycq}
\end{align}
\end{theorem}
\begin{proof}
Let $\theta=\Lambda(\rho_{AB},\sigma_{AB})$. If $\theta=0$ then $\rho_{AB}=\sigma_{AB}$ and there is nothing to prove, so assume $\theta>0$. Fix $\pi_{AB}>0$ and $\varepsilon\in(0,1)$, and let $\omega^\varepsilon_s$, $\Sig^\varepsilon_s$ and $h_\varepsilon$ be as in Lemma~\ref{lem:Hderiv}. Applying Proposition~\ref{prop:Vbounds} to the state $\omega^\varepsilon_s$ gives, for every $s\in(0,\theta)$,
\begin{align}
\sqrt{V_\alpha(\omega^\varepsilon_s\|\Sig^\varepsilon_s)} \;\le\; \frac{\sinh\big(|q|\arccosh d_A\big)}{|q|},
\end{align}
and the right-hand side does not depend on $s$. Together with~\eqref{eq:Hderiv},
\begin{align}
|h_\varepsilon'(s)| \;\le\; L\coloneqq 2\sqrt{1-\varepsilon}\,\frac{\sinh\big(|q|\arccosh d_A\big)}{|q|} \qquad\text{for all } s\in(0,\theta).
\end{align}
Since $h_\varepsilon$ is continuous on $[0,\theta]$ and differentiable on $(0,\theta)$, the mean value theorem gives $s^*\in(0,\theta)$ with $h_\varepsilon(\theta)-h_\varepsilon(0)=h_\varepsilon'(s^*)\,\theta$, and hence
\begin{align}
|h_\varepsilon(\theta)-h_\varepsilon(0)| \;\le\; L\,\theta.
\end{align}
As $\omega_0=\rho_{AB}$ and $\omega_\theta=\sigma_{AB}$, this reads
\begin{align}
\Big|\Hup(A|B)_{(1-\varepsilon)\sigma+\varepsilon\pi}-\Hup(A|B)_{(1-\varepsilon)\rho+\varepsilon\pi}\Big| \;\le\; 2\,\theta\,\sqrt{1-\varepsilon}\,\frac{\sinh\big(|q|\arccosh d_A\big)}{|q|}.
\end{align}
Letting $\varepsilon\to0$ and using the continuity of $\omega_{AB}\mapsto\Hup(A|B)_\omega$ (in terms of any norm) \cite{marwah2022uniform} gives~\eqref{eq:continuity}.

If $\rho_{AB}$ and $\sigma_{AB}$ are classical on $A$, we choose $\pi_{AB}$ classical on $A$ as well (e.g.\ $\pi_{AB}=\id_{AB}/(d_Ad_B)$), and replace the path of Lemma~\ref{lem:path} by that of Lemma~\ref{lem:pathcq}. Since the proof of Lemma~\ref{lem:Hderiv} only uses the properties of the path listed in Lemma~\ref{lem:path}, which the path of Lemma~\ref{lem:pathcq} shares, \eqref{eq:Hderiv} continues to hold. Moreover, $\omega^\varepsilon_s$ is now classical on $A$ for every $s$, so the classical-quantum bound~\eqref{eq:Vboundcq} of Proposition~\ref{prop:Vbounds} applies, and the same mean-value argument with $d_A$ replaced by $\sqrt{d_A}$ gives~\eqref{eq:continuitycq}.
\end{proof}
\begin{corollary}\label{cor:continuityvN}
Let $\rho_{AB},\sigma_{AB}\in S(AB)$. Then
\begin{align}
\big|H(A|B)_\rho-H(A|B)_\sigma\big| \;\le\; 2\,\Lambda(\rho_{AB},\sigma_{AB})\,\ln\!\Big(d_A+\sqrt{d_A^2-1}\Big). \label{eq:continuityvN}
\end{align}
If moreover $\rho_{AB}$ and $\sigma_{AB}$ are classical on $A$, 
% with respect to the same orthonormal basis of $\cH_A$, 
then
\begin{align}
\big|H(A|B)_\rho-H(A|B)_\sigma\big| \;\le\; 2\,\Lambda(\rho_{AB},\sigma_{AB})\,\ln\!\Big(\sqrt{d_A}+\sqrt{d_A-1}\Big). \label{eq:continuityvNcq}
\end{align}
\end{corollary}
\begin{proof}
Let $\alpha\to1$ in Theorem~\ref{thm:continuity}. Then $q\to0$, and $\sinh(|q|x)/|q|\to x$ for every $x\ge0$, so the right-hand sides of~\eqref{eq:continuity} and~\eqref{eq:continuitycq} tend to those of~\eqref{eq:continuityvN} and~\eqref{eq:continuityvNcq}. On the left-hand sides, $\Hup(A|B)_\omega\to H(A|B)_\omega$ as $\alpha\to1$ for every $\omega_{AB}\in S(AB)$ \cite{tomamichel2015quantum}. The closed forms follow from $\arccosh x=\ln\big(x+\sqrt{x^2-1}\big)$.
\end{proof}

Theorem~\ref{thm:intro} in the introduction is simply a summary of Theorem~\ref{thm:continuity} and Corollary~\ref{cor:continuityvN}.

\section*{Acknowledgements}
M.S.T.\ acknowledges support from the National Quantum Scholarships Scheme (Masters) scholarship. 
We thank the many people who have discussed and considered this question over the past years, including 
Mario Berta,
Andreas Bluhm,
Cheng Hao-Chung, 
Omar Fawzi,
Michael Jabbour,
Ludovico Lami,
Bartosz Regula,
Pavel Sekatski,
and
Marco Tomamichel.

\section*{Statement on AI use}
Large Language Models (LLMs) were used in the exploration of proof strategies, refining proofs, and preparation of the manuscript. The authors take responsibility for any errors. 

% \section{Comparison of results for Shannon entropy}
% \label{sec:compare_shannon}

\appendix
\section{Proof of auxiliary lemmas}\label{app:proofs}

\pathlemma*

\begin{proof}
By Uhlmann's theorem, and multiplying $|\Psi_\theta\rangle$ by a suitable phase, we get purifications $|\Psi_0\rangle_{ABC}$ of $\rho_{AB}$ and $|\Psi_\theta\rangle_{ABC}$ of $\sigma_{AB}$ such that
\begin{align}
\langle\Psi_0|\Psi_\theta\rangle = F(\rho_{AB},\sigma_{AB}) = \cos\theta .
\end{align}
Now let
\begin{align}
|\Psi_0^\perp\rangle_{ABC} \coloneqq \frac{|\Psi_\theta\rangle-\cos\theta\,|\Psi_0\rangle}{\big\|\,|\Psi_\theta\rangle-\cos\theta\,|\Psi_0\rangle\big\|}.
\end{align}
Then we have $\langle\Psi_0|\Psi_0^\perp\rangle=0$ by construction. Let
\begin{align}
|\Psi_s\rangle_{ABC} \coloneqq \cos s\,|\Psi_0\rangle + \sin s\,|\Psi_0^\perp\rangle \qquad\text{for } s\in[0,\theta].
\end{align}
This is a (manifestly smooth) path of vectors with
\begin{align}
\Tra_C\big[|\Psi_0\rangle\!\langle\Psi_0|_{ABC}\big]=\rho_{AB},\qquad
\Tra_C\big[|\Psi_\theta\rangle\!\langle\Psi_\theta|_{ABC}\big]=\sigma_{AB},\qquad
\|\Psi_s\|=1 \ \text{ for all } s\in[0,\theta].
\end{align}
Indeed, since $\langle\Psi_0|\Psi_\theta\rangle=\cos\theta$ is real and both vectors are unit vectors,
\begin{align}
\big\|\,|\Psi_\theta\rangle-\cos\theta\,|\Psi_0\rangle\big\|^2
&=\langle\Psi_\theta|\Psi_\theta\rangle-2\cos\theta\,\langle\Psi_0|\Psi_\theta\rangle+\cos^2\theta\,\langle\Psi_0|\Psi_0\rangle \notag\\
&=1-2\cos^2\theta+\cos^2\theta \notag\\
&=\sin^2\theta ,
\end{align}
which is non-zero since we may assume $\theta>0$ (if $\theta=0$ then $\rho_{AB}=\sigma_{AB}$ and there is nothing to prove). So $|\Psi_0^\perp\rangle=\big(|\Psi_\theta\rangle-\cos\theta\,|\Psi_0\rangle\big)/\sin\theta$, and $|\Psi_0\rangle,|\Psi_0^\perp\rangle$ are orthonormal. In particular, at $s=\theta$,
\begin{align}
\cos\theta\,|\Psi_0\rangle+\sin\theta\,|\Psi_0^\perp\rangle
&=\cos\theta\,|\Psi_0\rangle+|\Psi_\theta\rangle-\cos\theta\,|\Psi_0\rangle \notag\\
&=|\Psi_\theta\rangle ,
\end{align}
so the path ends at the purification of $\sigma_{AB}$, and $\|\Psi_s\|^2=\cos^2s+\sin^2s=1$ by orthonormality. We set $\omega_s\coloneqq\Tra_C\big[|\Psi_s\rangle\!\langle\Psi_s|_{ABC}\big]$.

Now let $|\Pi\rangle_{ABC}$ be a purification of $\pi_{AB}$, consider a qubit $C'$ with orthonormal basis $\{|0\rangle,|1\rangle\}$, and define
\begin{align}
|\Psi^\varepsilon_s\rangle_{ABCC'} \coloneqq \sqrt{1-\varepsilon}\,|\Psi_s\rangle_{ABC}|0\rangle_{C'} + \sqrt{\varepsilon}\,|\Pi\rangle_{ABC}|1\rangle_{C'} .
\end{align}
So
\begin{align}
\Tra_{CC'}\big[|\Psi^\varepsilon_s\rangle\!\langle\Psi^\varepsilon_s|\big] = (1-\varepsilon)\,\omega_s + \varepsilon\,\pi_{AB},
\end{align}
and the other properties hold with analogous checks.
\end{proof}
\pathlemmacq*
\begin{proof}
Let $\rho_{AB}$ and $\sigma_{AB}$ be classical on $A$ with respect to $\{|a\rangle\}_a$, and let $\omega_s$ and $|\Psi^\varepsilon_s\rangle_{ABCC'}$ be the path and purification constructed in the proof of Lemma~\ref{lem:path}. The states $\omega_s$ need not be classical on $A$, so we pinch with respect to the basis $\{|a\rangle\}_a$. Consider the pinching map
\begin{align}
\mathcal{P}(X_{AB})\coloneqq\sum_a\big(|a\rangle\!\langle a|_A\otimes\id_B\big)\,X_{AB}\,\big(|a\rangle\!\langle a|_A\otimes\id_B\big).
\end{align}
Then $\mathcal{P}(X_{AB})$ is classical on $A$ for every $X_{AB}$, and $\mathcal{P}$ fixes every state that is classical on $A$ with respect to $\{|a\rangle\}_a$. In particular $\mathcal{P}(\rho_{AB})=\rho_{AB}$, $\mathcal{P}(\sigma_{AB})=\sigma_{AB}$ and $\mathcal{P}(\pi_{AB})=\pi_{AB}$. A Stinespring dilation of $\mathcal{P}$ is given by the isometry $V:\cH_A\to\cH_A\otimes\cH_{A'}$, $V|a\rangle_A=|a\rangle_A|a\rangle_{A'}$, where $A'$ is a copy of $A$. Now let
\begin{align}
\widetilde\omega_s\coloneqq\mathcal{P}(\omega_s),\qquad
|\widetilde\Psi^\varepsilon_s\rangle_{ABCC'A'}\coloneqq(V\otimes\id_{BCC'})\,|\Psi^\varepsilon_s\rangle_{ABCC'} ,
\end{align}
with purifying system $CC'A'$. Then $\widetilde\omega_s$ is classical on $A$ for every $s$, $\widetilde\omega_0=\mathcal{P}(\rho_{AB})=\rho_{AB}$ and $\widetilde\omega_\theta=\mathcal{P}(\sigma_{AB})=\sigma_{AB}$. Both paths are smooth, since $\mathcal{P}$ and $V$ are linear and do not depend on $s$. So
\begin{align}
\Tra_{CC'A'}\big[|\widetilde\Psi^\varepsilon_s\rangle\!\langle\widetilde\Psi^\varepsilon_s|\big]
&=\Tra_{A'}\Big[(V\otimes\id_B)\,\Tra_{CC'}\big[|\Psi^\varepsilon_s\rangle\!\langle\Psi^\varepsilon_s|\big]\,(V\otimes\id_B)^\dagger\Big] \notag\\
&=\mathcal{P}\big((1-\varepsilon)\,\omega_s+\varepsilon\,\pi_{AB}\big) \notag\\
&=(1-\varepsilon)\,\widetilde\omega_s+\varepsilon\,\pi_{AB}.
\end{align}
Finally, $V\otimes\id_{BCC'}$ is an isometry that does not depend on $s$, so $|\dot{\widetilde\Psi}{}^\varepsilon_s\rangle=(V\otimes\id_{BCC'})|\dot\Psi^\varepsilon_s\rangle$, and since isometries preserve inner products,
\begin{align}
\|\dot{\widetilde\Psi}{}^\varepsilon_s\|=\|\dot\Psi^\varepsilon_s\|=\sqrt{1-\varepsilon},\qquad
\langle\widetilde\Psi^\varepsilon_s|\dot{\widetilde\Psi}{}^\varepsilon_s\rangle=\langle\Psi^\varepsilon_s|\dot\Psi^\varepsilon_s\rangle=0. \qedhere
\end{align}
\end{proof}

\momentlemma*
\begin{proof}
Since $\cosh$ is even, $f_q=f_{|q|}$, and $f_0\equiv1$, so we may assume $q\in(0,1]$. Substituting $x=\cosh t$ with $t>0$, we have $f_q(\cosh t)=\cosh(qt)$, and hence
\begin{align}
\frac{d f_q}{dx}=\frac{q\sinh(qt)}{\sinh t}\;\ge\;0,
\qquad
\frac{d^2 f_q}{dx^2}=\frac{q\big(q\cosh(qt)\sinh t-\sinh(qt)\cosh t\big)}{\sinh^3 t}.
\end{align}
Since $\tanh$ is concave on $[0,\infty)$ with $\tanh(0)=0$, we have $\tanh(qt)\ge q\tanh t$, i.e.\ $\sinh(qt)\cosh t\ge q\cosh(qt)\sinh t$, so the second derivative is non-positive. Thus $f_q$ is increasing and concave on $[1,\infty)$. Finally, $\arccosh(\cosh Z)=|Z|$, so $\cosh(qZ)=\cosh(q|Z|)=f_q(\cosh Z)$, and Jensen's inequality for the concave function $f_q$ finishes the proof. 
\end{proof}
\section{Continuity of ``quasi-entropy-type'' quantities}
\label{app:quasient}
In this section we discuss continuity bounds for a particular exponentiation of the {\Renyi} entropy, namely (for $q=\frac{\alpha-1}{\alpha}$):
\begin{equation}
  R_\alpha(A|B)_\omega \coloneqq e^{-q\widetilde{H}^\uparrow_\alpha(A|B)_\omega}.
\end{equation}
This quantity is independent of the base used to define logarithms and entropies (as long as the base is used consistently), i.e.~it can be equivalently written as
\begin{equation}
  R_\alpha(A|B)_\omega = 2^{-q\widetilde{H}^\uparrow_{\alpha, \text{base-$2$}}(A|B)_\omega},
\end{equation}
where $\widetilde{H}^\uparrow_{\alpha, \text{base-$2$}}$ denotes the base-$2$ entropy.

This quantity often arises in QKD applications; see e.g.~\cite{kamin2025renyi}. Qualitatively, it is roughly similar to what is sometimes called a ``quasi-entropy'', namely the argument of the logarithm in the {\Renyi} divergences (after specializing to conditional entropies). However, we have a slightly different factor $q=\frac{\alpha-1}{\alpha}$ in the exponent, which roughly arises from the fact that we considered the ``uparrow'' version of the entropy.

 The overall strategy is the same as for the entropic continuity bound: we purify, follow the path of Lemma~\ref{lem:path} (or Lemma~\ref{lem:pathcq}), and bound the derivative along it. The difference is that we do not bound $V_\alpha$ by a constant. Instead, Lemma~\ref{lem:Rvariance} below keeps track of the term $M(-q)$ that is discarded in the proof of Lemma~\ref{lem:Vgeneral}. Since $M(-q)=1/R_\alpha(A|B)_\omega$, this gives a bound on $V_\alpha$ that depends on the current value of $R_\alpha(A|B)_\omega$, and the resulting differential inequality can be integrated exactly.

Throughout this section, we take $\alpha\in(\frac12,1)\cup(1,\infty)$, $q=\frac{\alpha-1}{\alpha}$, $\chi=\frac12$ if the states under consideration are classical on $A$ 
% with respect to the same orthonormal basis of $\cH_A$ 
and $\chi=1$ otherwise, and
\begin{equation}\label{eq:cdefn}
  c \coloneqq \cosh\big(q\arccosh d_A^{\chi}\big),
  \qquad\text{so that}\qquad
  c\pm\sqrt{c^2-1}=e^{\pm|q|\arccosh d_A^{\chi}} .
\end{equation}
Since $q\neq0$ and $d_A^\chi>1$, we have $c>1$. To avoid some singularities in the subsequent expressions, we also restrict to $d_A \geq 2$.

\begin{lemma}\label{lem:Rvariance}
Let $\omega_{AB}\in S(AB)$ with $\omega_{AB}>0$, and let $\Sig\coloneqq\id_A\otimes\sigmaopt$, where $\sigmaopt$ is the optimizer of Lemma~\ref{lem:opt} for $\omega_{AB}$. Then (for $\chi=1/2$ if $\omega$ is classical on $A$, and $\chi=1$ otherwise):
\begin{equation}
  q^2\,d_A^{2q(1-\chi)}\,R_\alpha(A|B)_\omega^2\;V_\alpha(\omega_{AB}\|\Sig)
  \;\le\;(c^2-1)-\big(d_A^{q(1-\chi)}R_\alpha(A|B)_\omega-c\big)^2 ,
\end{equation}
for $c$ as defined in~\eqref{eq:cdefn}.
\end{lemma}
\begin{proof}
Since $\omega_{AB}>0$ we have $\sigmaopt>0$. Let $X$, $Q$, $M$, $P$ and $Z$ be as in Lemma~\ref{lem:Vgeneral} and its proof, for the state $\omega_{AB}$. By the proof of Lemma~\ref{lem:product}, $M(-q)=Q^{-(1-q)}=Q^{-1/\alpha}$, and by~\eqref{eq:Hupopt}, $Q=e^{(\alpha-1)\Dsand(\omega_{AB}\|\Sig)}=e^{-(\alpha-1)\Hup(A|B)_\omega}$. Hence
\begin{equation}
  M(-q)=e^{q\Hup(A|B)_\omega}=\frac{1}{R_\alpha(A|B)_\omega}.
\end{equation}
Let $\mu\coloneqq(\chi-1)\ln d_A$. By the proof of Proposition~\ref{prop:Vbounds}, $M(-1)\le d_A$ and $M(1)\le d_A$, and moreover $M(1)\le1$ if $\omega_{AB}$ is classical on $A$. In both cases,
\begin{equation}
  \EE_P\big[\cosh(Z-\mu)\big]=\tfrac12\big(e^{-\mu}M(1)+e^{\mu}M(-1)\big)\le d_A^{\chi}.
\end{equation}
Since $e^{-q\mu}M(q)=\EE_P[e^{q(Z-\mu)}]$ and $e^{q\mu}M(-q)=\EE_P[e^{-q(Z-\mu)}]$, Lemma~\ref{lem:moment} and the monotonicity of $x\mapsto\cosh(q\arccosh x)$ give
\begin{align}
  e^{-q\mu}M(q)+e^{q\mu}M(-q)
  &=2\,\EE_P\big[\cosh\big(q(Z-\mu)\big)\big] \notag\\
  &\le2\cosh\Big(q\arccosh\EE_P\big[\cosh(Z-\mu)\big]\Big) \notag\\
  &\le2\cosh\big(q\arccosh d_A^{\chi}\big) \notag\\
  &=2c .
\end{align}
Hence $e^{-q\mu}M(q)\le2c-e^{q\mu}M(-q)$. Multiplying by $e^{q\mu}M(-q)\ge0$ and using Lemma~\ref{lem:product},
\begin{equation}
  \Tra\big[\omega_{AB}Y_\alpha(\omega_{AB}\|\Sig)^2\big]=M(q)M(-q)\le e^{q\mu}M(-q)\big(2c-e^{q\mu}M(-q)\big).
\end{equation}
Now $e^{q\mu}M(-q)=1/\big(d_A^{q(1-\chi)}R_\alpha(A|B)_\omega\big)$ and $q^2V_\alpha(\omega_{AB}\|\Sig)=\Tra\big[\omega_{AB}Y_\alpha(\omega_{AB}\|\Sig)^2\big]-1$. Multiplying by $d_A^{2q(1-\chi)}R_\alpha(A|B)_\omega^2$ therefore gives
\begin{align}
  q^2\,d_A^{2q(1-\chi)}R_\alpha(A|B)_\omega^2\,V_\alpha(\omega_{AB}\|\Sig)
  &\le2c\,d_A^{q(1-\chi)}R_\alpha(A|B)_\omega-1-d_A^{2q(1-\chi)}R_\alpha(A|B)_\omega^2 \notag\\
  &=(c^2-1)-\big(d_A^{q(1-\chi)}R_\alpha(A|B)_\omega-c\big)^2 . \qedhere
\end{align}
\end{proof}

\begin{proposition}\label{prop:Rarcsin}
Let $\rho_{AB},\sigma_{AB}\in S(AB)$. Then (for $\chi=1/2$ if both states are classical on $A$, and $\chi=1$ otherwise):
\begin{equation}\label{eq:Rarcsin}
  \left|\arcsin\frac{d_A^{q(1-\chi)}R_\alpha(A|B)_\rho-c}{\sqrt{c^2-1}}-\arcsin\frac{d_A^{q(1-\chi)}R_\alpha(A|B)_\sigma-c}{\sqrt{c^2-1}}\right|
  \;\le\;2\Lambda(\rho_{AB},\sigma_{AB}),
\end{equation}
for $c$ as defined in~\eqref{eq:cdefn}.
\end{proposition}
\begin{proof}
We first check that the arcsine is differentiable at every value that occurs. We have $-\ln d_A\le\Hup(A|B)_\omega\le\ln d_A$, and $\Hup(A|B)_\omega\ge0$ if $\omega_{AB}$ is classical on $A$~\cite{tomamichel2015quantum}. Hence $d_A^{q(1-\chi)}R_\alpha(A|B)_\omega\in\big[d_A^{-|q|\chi},d_A^{|q|\chi}\big]$ in both cases. Since $\arccosh x=\ln\big(x+\sqrt{x^2-1}\big)>\ln x$ for $x>1$, this interval lies strictly inside $\big(e^{-|q|\arccosh d_A^\chi},e^{|q|\arccosh d_A^\chi}\big)=\big(c-\sqrt{c^2-1},\,c+\sqrt{c^2-1}\big)$.

Let $\theta\coloneqq\Lambda(\rho_{AB},\sigma_{AB})$. If $\theta=0$ there is nothing to prove, so assume $\theta>0$. Fix $\pi_{AB}>0$ and $\varepsilon\in(0,1)$, and let $\omega^\varepsilon_s$, $\Sig^\varepsilon_s=\id_A\otimes\sigma^\varepsilon_s$ and $h_\varepsilon$ be as in Lemma~\ref{lem:Hderiv}. If $\rho,\sigma$ are both classical on $A$, take $\pi_{AB}=\id_{AB}/(d_Ad_B)$ and the path of Lemma~\ref{lem:pathcq} instead of that of Lemma~\ref{lem:path}, so that every $\omega^\varepsilon_s$ is classical on $A$; as noted in the proof of Theorem~\ref{thm:continuity}, the conclusions of Lemma~\ref{lem:Hderiv} continue to hold for this path.

By definition, $R_\alpha(A|B)_{\omega^\varepsilon_s}=e^{-q\,h_\varepsilon(s)}$. Since $h_\varepsilon$ is continuous on $[0,\theta]$ and differentiable on $(0,\theta)$, so is $s\mapsto R_\alpha(A|B)_{\omega^\varepsilon_s}$, and by the chain rule and~\eqref{eq:Hderiv},
\begin{align}
  \Big|\frac{d}{ds}R_\alpha(A|B)_{\omega^\varepsilon_s}\Big|
  &=|q|\,R_\alpha(A|B)_{\omega^\varepsilon_s}\,|h_\varepsilon'(s)| \notag\\
  &\le2\sqrt{1-\varepsilon}\;|q|\,R_\alpha(A|B)_{\omega^\varepsilon_s}\sqrt{V_\alpha(\omega^\varepsilon_s\|\Sig^\varepsilon_s)}
\end{align}
for all $s\in(0,\theta)$. Each $\omega^\varepsilon_s$ is full rank, 
% and classical on $A$ if the original states $\rho,\sigma$ were so, 
thus Lemma~\ref{lem:Rvariance} applies.
Moreover, if the original states $\rho,\sigma$ were classical on $A$ then each $\omega^\varepsilon_s$ is also classical on $A$, so in that case we can validly use the $\chi=1/2$ case in that lemma. (If $\rho,\sigma$ are not both classical on $A$, the following analysis should be understood as always taking $\chi=1$ even if some $\omega^\varepsilon_s$ are classical on $A$, since this is a valid relaxation of the $\chi=1/2$ bound.) 
Multiplying by $d_A^{q(1-\chi)}$,
\begin{equation}
  \Big|\frac{d}{ds}\Big(d_A^{q(1-\chi)}R_\alpha(A|B)_{\omega^\varepsilon_s}\Big)\Big|
  \le2\sqrt{1-\varepsilon}\,\sqrt{(c^2-1)-\big(d_A^{q(1-\chi)}R_\alpha(A|B)_{\omega^\varepsilon_s}-c\big)^2} .
\end{equation}
By the chain rule, it follows that
\begin{align}
  \left|\frac{d}{ds}\arcsin\frac{d_A^{q(1-\chi)}R_\alpha(A|B)_{\omega^\varepsilon_s}-c}{\sqrt{c^2-1}}\right|
  &=\frac{\Big|\frac{d}{ds}\Big(d_A^{q(1-\chi)}R_\alpha(A|B)_{\omega^\varepsilon_s}\Big)\Big|}{\sqrt{(c^2-1)-\big(d_A^{q(1-\chi)}R_\alpha(A|B)_{\omega^\varepsilon_s}-c\big)^2}} \notag\\
  &\le2\sqrt{1-\varepsilon}
\end{align}
for all $s\in(0,\theta)$. The mean value theorem bounds the change of the arcsine between $s=0$ and $s=\theta$ by $2\sqrt{1-\varepsilon}\,\theta$. Letting $\varepsilon\to0$ and using the continuity of $\omega_{AB}\mapsto\Hup(A|B)_\omega$ as in the proof of Theorem~\ref{thm:continuity} gives~\eqref{eq:Rarcsin}.
\end{proof}

We now present two corollaries that are consequences of Proposition~\ref{prop:Rarcsin}, which should be useful for QKD applications. For the first, Corollary~\ref{cor:Rstate}, we upper bound $R_\alpha(A|B)_\rho$ in terms of a concave nondecreasing function of $R_\alpha(A|B)_\sigma$ (given some constant upper bound $\theta$ on the angular distance between the states). For the purposes of QKD numerical optimizations, this means that e.g.~a maximization problem where the objective function is $R_\alpha(A|B)_\rho$ (or a positive multiple thereof) can be relaxed to one where the objective function is concave in $R_\alpha(A|B)_\sigma$. This would, in principle, preserve the concavity/convexity properties typically needed for these numerical optimization methods. However, there is the drawback that this concave function is slightly elaborate, and may not be supported by standard solvers. Thus as an alternative, we present also Corollary~\ref{cor:Radd}, which simply gives a bound on the ``additive'' difference $\left|R_\alpha(A|B)_\rho - R_\alpha(A|B)_\sigma\right|$, which should be easier to implement numerically.

\begin{corollary}\label{cor:Rstate}
Under the assumptions of Proposition~\ref{prop:Rarcsin}, let $\Lambda(\rho_{AB},\sigma_{AB})\le\theta\le\frac\pi2$. Then
% \begin{equation}
%   d_A^{q(1-\chi)}R_\alpha(A|B)_\rho
%   \le\min\left\{d_A^{|q|\chi},\;c+\sqrt{c^2-1}\,\sin\left[\min\left\{\frac\pi2,\;\arcsin\frac{d_A^{q(1-\chi)}R_\alpha(A|B)_\sigma-c}{\sqrt{c^2-1}}+2\theta\right\}\right]\right\}.
% \end{equation}
\begin{equation}
R_\alpha(A|B)_\rho \leq F_\theta(R_\alpha(A|B)_\sigma),
\end{equation}
where $F_\theta$ is the following concave nondecreasing function on $\left[
d_A^{q(\chi-1)-|q|\chi},
d_A^{q(\chi-1)+|q|\chi}
\right]$:
\begin{equation}\label{eq:concavebnd}
F_\theta(r)
\coloneqq d_A^{q(\chi-1)}
\min\Bigg\{d_A^{|q|\chi}, 
c+\sqrt{c^2-1}\,
\sin\!\left[\min\!\left\{
\frac{\pi}{2},
\arcsin\!\left(
\frac{d_A^{q(1-\chi)}r-c}{\sqrt{c^2-1}}
\right)+2\theta
\right\}\right]\Bigg\},
\end{equation}
for $c$ as defined in~\eqref{eq:cdefn}.
\end{corollary}
\begin{proof}
By Proposition~\ref{prop:Rarcsin} and $\Lambda(\rho_{AB},\sigma_{AB})\le\theta$, the arcsine evaluated at $\rho_{AB}$ is at most the arcsine evaluated at $\sigma_{AB}$ plus $2\theta$, and it is also at most $\frac\pi2$. Applying the map $x\mapsto c+\sqrt{c^2-1}\sin x$, which is increasing on $[-\frac\pi2,\frac\pi2]$ and inverts the arcsine, and using $d_A^{q(1-\chi)}R_\alpha(A|B)_\rho\le d_A^{|q|\chi}$ from the proof of Proposition~\ref{prop:Rarcsin}, gives the claimed bound.

To prove the claimed properties of $F_\theta$, we first cover the trivial ``endpoint'' values: $F_0(r)=r$ and $F_{\pi/2}(r)=d_A^{q(\chi-1)+|q|\chi}$ trivially have the claimed properties. Then consider the following function for 
 $0<\theta<\pi/2$ and $x\in[-1,1]$:
\begin{equation}
g_\theta(x)
\coloneqq
% \sin\!\left(\min\{\pi/2,\arcsin x+2\theta\}\right)=
\begin{cases}
x\cos(2\theta)+\sqrt{1-x^2}\sin(2\theta),
& x\le\cos(2\theta),\\
1, & x\ge\cos(2\theta).
\end{cases}
\end{equation}
Note that
\begin{equation}
F_\theta(r)
= d_A^{q(\chi-1)}
\min\left\{
d_A^{|q|\chi},
c+\sqrt{c^2-1}\,
g_\theta\!\left(
\frac{d_A^{q(1-\chi)}r-c}{\sqrt{c^2-1}}
\right)
\right\},
\end{equation}
i.e.~$F_\theta$ is the minimum of a constant and an affine transformation of $g_\theta$ (with positive coefficients). Thus it suffices to prove that $g_\theta$ is concave and nondecreasing, after which $F_\theta$ inherits those properties. 

For the $x\le\cos(2\theta)$ branch in the definition of $g_\theta$, we have
\begin{equation}
g_\theta''(x)
=-\frac{\sin(2\theta)}{(1-x^2)^{3/2}}\le0,
\qquad
g_\theta'\bigl(\cos(2\theta)\bigr)=0.
\end{equation}
Thus the derivative is nonnegative and nonincreasing, and
joins continuously to the derivative of the $x\ge\cos(2\theta)$
branch (which is just zero, since $g_\theta$ is constant there). Hence $g_\theta$ is concave and nondecreasing.
\end{proof}

\begin{corollary}\label{cor:Radd}
Under the assumptions of Proposition~\ref{prop:Rarcsin},
\begin{equation}
  \big|R_\alpha(A|B)_\rho-R_\alpha(A|B)_\sigma\big|
  \le2\,d_A^{q(\chi-1)}\sqrt{c^2-1}\,\sqrt{1-F(\rho_{AB},\sigma_{AB})^2},
\end{equation}
for $c$ as defined in~\eqref{eq:cdefn}.
\end{corollary}
\begin{proof}
Both $d_A^{q(1-\chi)}R_\alpha(A|B)_\rho$ and $d_A^{q(1-\chi)}R_\alpha(A|B)_\sigma$ are of the form $c+\sqrt{c^2-1}\sin x$, with $x$ the respective arcsine in~\eqref{eq:Rarcsin}. By the identity $\sin x-\sin y=2\cos\frac{x+y}2\sin\frac{x-y}2$, their difference is at most $2\sqrt{c^2-1}\,\big|\sin\frac{x-y}2\big|$ in absolute value. By Proposition~\ref{prop:Rarcsin}, $\frac{|x-y|}2\le\Lambda(\rho_{AB},\sigma_{AB})\le\frac\pi2$, so this is at most $2\sqrt{c^2-1}\,\sin\Lambda(\rho_{AB},\sigma_{AB})=2\sqrt{c^2-1}\sqrt{1-F(\rho_{AB},\sigma_{AB})^2}$. Dividing by $d_A^{q(1-\chi)}$ gives the claim.
\end{proof}

\begin{remark}
While the above two results have been presented for ease of applicability, strictly speaking they might not be the tightest possible way to apply this general proof technique to QKD numerical work. For instance, when it comes to the additive form in Corollary~\ref{cor:Radd}: rather than $R_\alpha(A|B)$ specifically, one could consider the continuity of $e^{-q'\widetilde{H}^\uparrow_\alpha(A|B)}$ for some other value $q'$ (e.g.~$q'=\alpha-1$ would be a natural guess). Some preliminary analysis indicates that neither approach is strictly better, and a case-by-case analysis may generally be needed. However, we defer such a fine-grained analysis to future work.
\end{remark}

\section{Limiting arguments}
\label{app:limits}

While the limiting arguments are basically routine, we choose to minimize the extent of dependency on infinite-dimensional results by \emph{first} extending to the ``endpoint'' values of $\alpha$, \emph{then} extending to countably-infinite-dimensional $B$. 

Specifically, to handle $\alpha=\infty$ we simply observe that \emph{by definition} the {\Renyi} divergence at that value $\alpha=\infty$ is given by that limit, and our continuity bound formulas converge to well-defined values in that limit as well. 
There is the minor subtlety that $\Hup$ involves an optimization over the divergence, but the interchange between this optimization and the $\alpha\to\infty$ limit is established by~\cite{mosonyi2023strong}.
For $\alpha=1/2$, we invoke the fact that the {\Renyi} divergences are continuous in $\alpha$ at that value~\cite[Chapter~4.1.3]{tomamichel2015quantum}, so we can perform an analogous argument but also literally substitute $\alpha=1/2$ at the end.

To handle (countably) infinite-dimensional $B$, the intuition is also simple: we approximate  $(\rho,\sigma)$ with a sequence of states $(\rho_n, \sigma_n)$ that have finite support on $B$ (preserving the classical structure on $A$ if necessary), in such a way that
\begin{equation}
\begin{gathered}\label{eq:convergecriteria}
\widetilde H_\alpha^\uparrow(A|B)_{\rho_n}
\to
\widetilde H_\alpha^\uparrow(A|B)_\rho, \qquad
\widetilde H_\alpha^\uparrow(A|B)_{\sigma_n}
\to
\widetilde H_\alpha^\uparrow(A|B)_\sigma, \\
\Lambda(\rho_n,\sigma_n) \to \Lambda(\rho,\sigma).
\end{gathered}
\end{equation}
Since each such pair $(\rho_n, \sigma_n)$ has finite support on $B$, all our bounds apply to them, e.g.~for Theorem~\ref{thm:intro} we get
\begin{align}
\Big|\Hup(A|B)_{\rho_n}-\Hup(A|B)_{\sigma_n}\Big| \;\le\; 2\,\Lambda(\rho_{n},\sigma_{n})\,\frac{\sinh\big(|q|\arccosh (d_A^\chi)\big)}{|q|},
\end{align}
and then taking limits on both sides of the equation yields the desired claim for the original states $(\rho,\sigma)$ with infinite-dimensional $B$ (implicitly using the fact that $|x-y|$ is a continuous function, and the coefficient ${\sinh\big(|q|\arccosh (d_A^\chi)\big)}/{|q|}$ on the RHS has no dependence on the states $(\rho_n, \sigma_n)$ except via the dimension $d_A$, which remains constant throughout). A similar argument holds for our other continuity bounds.

The main technicality is simply in defining $\Hup(A|B)$ for countably-infinite-dimensional $B$ in the first place, and then ensuring that this desired sequence of approximations indeed exists. This was already basically established in previous works, although a bit of extra fine-tuning is required to obtain the exact version we want here. We thus list a few options below.

\subsection{Construction of approximations}

To define $\Hup(A|B)$ for countably infinite dimension, we follow~\cite{mosonyi2023strong,hiai2023quantum}:
they provide several equivalent definitions of $\Dsand$ (one of them being, in fact, precisely a limiting definition), then defines $\Hup$ via the same optimization as in Definition~\ref{def:Hup} (albeit with infimum instead of minimum). The following discussions apply to this choice of definition.

For $\alpha\in(1,\infty)$, one approach is to use a recent work~\cite{kamin2026unconditional}, which gave a continuity bound for $\Hup$ (depending only on $d_A$) in terms of trace distance for infinite-dimensional $B$, by building on earlier convergence results from~\cite{mosonyi2023strong}. This implies continuity in terms of angular distance as well, via Fuchs--van de Graaf: while the resulting bound is rather loose, it is sufficient to simply establish continuity, which is all we need for the remainder of this argument. With this, we can simply pick any sequence of states $(\rho_n, \sigma_n)$ (with finite support on $B$) converging in angular distance to $(\rho,\sigma)$, and invoke continuity of $\Hup$ to say that their entropies $\Hup$ converge as required in~\eqref{eq:convergecriteria}. Their angular distance also converges as desired, by a standard triangle-inequality argument for metrics, i.e.
\begin{equation}
\begin{gathered}
\Lambda(\rho_n,\sigma_n) \leq \Lambda(\rho_n,\rho) + \Lambda(\rho,\sigma) + \Lambda(\sigma,\sigma_n) ,\\
\Lambda(\rho,\sigma) \leq \Lambda(\rho,\rho_n) + \Lambda(\rho_n,\sigma_n) + \Lambda(\sigma_n,\sigma) 
% \quad\implies\quad \Lambda(\rho_n,\sigma_n) \geq \Lambda(\rho,\sigma) - \Lambda(\rho,\rho_n) -\Lambda(\sigma_n,\sigma)
,
\end{gathered}
\end{equation}
so 
\begin{equation}
\begin{gathered}
\left|\Lambda(\rho_n,\sigma_n) - \Lambda(\rho,\sigma)\right| \leq \Lambda(\rho_n,\rho) + \Lambda(\sigma,\sigma_n) \to 0,
\end{gathered}
\end{equation}
as desired.

Another approach for $\alpha\in(1,\infty)$ is to use an earlier work~\cite{beigi2023operator}, which also established continuity in trace distance: while it did not explicitly claim its result holds for countably-infinite-dimensional $B$, an inspection of the proof shows that it applies directly to that setting. Then we can apply the same argument as described above.

For $\alpha=1$, the required convergence follows from
\cite[Proposition~2]{kuznetsova2011conditional}, since $A$ is
finite dimensional. For $\alpha=\infty$, the required convergence
follows from \cite[Proposition~1]{furrer2011min} by considering
normalized finite-dimensional projections.

For $1/2\leq\alpha<1$, we need to do slightly more work on top of existing results from~\cite{hiai2023quantum}, basically to ensure the optimization in $\Hup$ interacts correctly with their convergence results for $\Dsand$.
Let
$P_n\uparrow I_B$ be finite-rank projectors converging strongly to $I_B$, and extend them onto $A$ via $\Pi_n=I_A\otimes P_n$. Now for any $\rho$, we construct a sequence of normalized projected versions $\rho_n$ via
\begin{equation}
\rho_n=\frac{\Pi_n\rho\Pi_n}{p_n} , \quad \text{ where }
p_n \coloneqq \operatorname{Tr}(\Pi_n\rho).
\end{equation}
Strictly speaking, there may be finitely many  terms with $p_n=0$ so the above formula is technically ill-defined; let us restrict our attention to only values of $n$ beyond which all $p_n$ are nonzero.
For these states $\rho_n$, every state $\sigma_B$
supported on $P_n$ satisfies
\begin{gather}
(I_A\otimes\sigma_B)^{\frac{1-\alpha}{2\alpha}}\rho\,(I_A\otimes\sigma_B)^{\frac{1-\alpha}{2\alpha}} = p_n (I_A\otimes\sigma_B)^{\frac{1-\alpha}{2\alpha}}\rho_n\,(I_A\otimes\sigma_B)^{\frac{1-\alpha}{2\alpha}}\\
% \Tra\Big[\Big((I_A\otimes\sigma_B)^{\frac{1-\alpha}{2\alpha}}\rho\,(I_A\otimes\sigma_B)^{\frac{1-\alpha}{2\alpha}}\Big)^{\alpha}\Big] = p_n^\alpha \,\Tra\Big[\Big((I_A\otimes\sigma_B)^{\frac{1-\alpha}{2\alpha}}\rho_n\,(I_A\otimes\sigma_B)^{\frac{1-\alpha}{2\alpha}}\Big)^{\alpha}\Big] \\
\implies \widetilde D_\alpha(\rho\Vert I_A\otimes\sigma_B)
=
\widetilde D_\alpha(\rho_n\Vert I_A\otimes\sigma_B)
+\frac{\alpha}{\alpha-1}\ln p_n.\label{eq:renyi-projection-reference-below-one}
\end{gather}
Thus for each $n$, letting $S_n(B)$ denote the set of states on $B$ supported on $P_n$, we have
\begin{gather}
\inf_{\sigma_B \in S_n(B)} \widetilde D_\alpha(\rho\Vert I_A\otimes\sigma_B)
=
-\Hup(A|B)_{\rho_n}
+\frac{\alpha}{\alpha-1}\ln p_n, \label{eq:projected_entropy}
\end{gather}
where in simplifying the RHS we made the observation that second-argument optimization in the definition of $\Hup(A|B)_{\rho_n}$ can be restricted to $S_n(B)$ without loss of generality. 

Note the LHS need not equal $-\Hup(A|B)_{\rho}$ because the infimum is over a restricted subset, not the whole of $S(B)$. However, in the $n\to\infty$ limit it does converge to $-\Hup(A|B)_{\rho} 
% = \inf_{\sigma_B \in S(B)} \widetilde D_\alpha(\rho\Vert I_A\otimes\sigma_B)
$. This follows from the fact that for any $\sigma_B \in S(B)$ (possibly with infinite support), its normalized projected versions
$\sigma_{B,n}=P_n\sigma_BP_n/\operatorname{Tr}(P_n\sigma_B)$ (where again the denominator will be non-zero for all sufficiently large $n$) satisfy
$\widetilde D_\alpha(\rho\Vert I_A\otimes\sigma_{B,n})
\to\widetilde D_\alpha(\rho\Vert I_A\otimes\sigma_B)$, by the convergence result in~\cite[Proposition~3.4]{hiai2023quantum} together with an application of \eqref{eq:renyi-projection-reference-below-one}.
Consequently, the infima over $S_n(B)$ converge to the infimum over $S(B)$. This fact, together with~\eqref{eq:projected_entropy} and the fact that $p_n\to1$, yields the desired property
$\widetilde H_\alpha^\uparrow(A|B)_{\rho_n}
\to\widetilde H_\alpha^\uparrow(A|B)_\rho$.

\bibliography{References.bib}

\end{document}